\documentclass[preprint,12pt]{elsarticle}
\usepackage{amsmath,amsthm,amssymb,amsfonts,mathdots,graphicx,algorithm,algpseudocode}
\usepackage{a4wide}
\usepackage{bm}
\usepackage{tikz}
\usetikzlibrary{decorations.markings}

\newcommand{\inp}{\textit{inp}}
\newcommand{\BigO}{\mathop{\text{O}}}

\newcommand{\om}{\mathop{\mathrm{\omega}}}
\newcommand{\Th}{\mathop{\mathrm{\Theta}}}

\newcommand{\defeq}{\stackrel{\text{def}}{=}}

\newcommand{\minus}{\mathop{\tikz{
\coordinate (P1) at (0,0.1);
\coordinate (P2) at (0.2,0.1);
\coordinate (P3) at (0.1,0.18);
\coordinate (P4) at (0,0);
\coordinate (P5) at (0.2,0.2);
\draw[fill=white, color=white] (P4) rectangle (P5);
\draw[-] (P1) -- (P2);
\fill (P3) circle(1pt);
}}}
\newtheorem{thm}{Theorem}

\newtheorem{cor}{Corollary}[thm]

\theoremstyle{definition}
\newtheorem{defi}{Definition}

\begin{document}

\begin{frontmatter}
\title{Arbitrary Sequence RAMs}
\author[monash]{Michael Brand}
\ead{michael.brand@alumni.weizmann.ac.il}
\address[monash]{Faculty of IT, Monash University, Clayton, VIC 3800, Australia}

\begin{abstract}
It is known that in some cases a Random Access Machine (RAM) benefits from
having an additional input that is an arbitrary number, satisfying only the
criterion of being sufficiently large. This is known as the ARAM model.
We introduce a new type of RAM, which we refer to as
the Arbitrary Sequence RAM (ASRAM), that generalises the ARAM by allowing
the generation of additional arbitrary large numbers at will during
execution time. We characterise the power contribution of this ability under
several RAM variants.

In particular, we demonstrate that an arithmetic ASRAM is more powerful than
an arithmetic ARAM, that a sufficiently equipped ASRAM can recognise any
language in the arithmetic hierarchy in constant time (and more, if it is
given more time), and that, on the other hand, in some cases the ASRAM is no
more powerful than its underlying RAM.
\end{abstract}


\end{frontmatter}

\section{Introduction}
The Random Access Machine, or RAM,
(see \cite{Aho:Algorithms} for a formal definition)
is a computational model that affords all that we expect from a modern computer
in terms of flow control (loops, conditional jump instructions, etc.) and
access to variables (direct and indirect addressing).
It is denoted by $\text{RAM}[\textit{op\/}]$, where $\textit{op}$ is the set
of operations that are assumed to be executable by the RAM
in a single unit of time each. A comparator for equality is also
assumed to be available, and this also executes in a single unit of
time. The variables (or \emph{registers}) of an integer RAM contain nonnegative
integers and are also indexable by addresses that are nonnegative integers.

To discuss the power of RAMs, let us consider RAMs as calculating functions.
We initialise the RAM by storing the input value, $\inp$, in the RAM's
$R[0]$ register (setting all other registers to zero), and the output of the
function is taken to be the value of $R[0]$ at termination time.
This definition can be extended to functions receiving any fixed
number of inputs.
Alternatively, RAMs can be discussed as language acceptors, where
$\inp$ is taken to be in the language if and only if the return value
is non-zero. Traditionally, when viewing the RAM as an acceptor, non-termination
is taken to mean rejection of the input. By contrast, when viewing the RAM as a
function calculator, non-termination is usually taken to mean that the RAM
calculates a partial function, rather than a function. The two frameworks can be
unified by arbitrarily taking a non-terminating computation in a
function-calculating RAM to yield an output of zero.\footnote{All RAMs
discussed in this paper are guaranteed to terminate in finite time, so handling
the case of non-termination is never an issue.}

Ben-Amram and Galil \cite{Galil:Shift} write
``The RAM is intended to model what
we are used to in conventional programming, idealized in order to be better
accessible for theoretical study.'' 
However, in practice, the RAM's ability to manipulate very large numbers in
constant time has been shown to reduce algorithmic complexities beyond what
is usually considered ``reasonable''. 
For example, it was shown regarding many RAMs
working with fairly limited instruction sets that they are able to
recognise any PSPACE problem in deterministic polynomial time
\cite{Schonhage:rams, Simon:Multiplication, Bertoni:pTime_RAM, Pratt:VMs}, and
a unit-cost RAM equipped only with arithmetic operations, Boolean operations and
bit shifts can, in fact, recognise in constant time any language that is
recognised by a TM in time and/or space constrained by \emph{any} elementary
function of the input size \cite{Brand:ALNs}.

In most cases (e.g., \cite{Mansour:GCD, Mansour:sqrt, Mansour:floor1,
Mansour:floor2, Shamir:factoring, Paul:sorting}), the large integers to be
efficiently manipulated by the RAM are generated to precise values that are
conducive to the computation at hand. However, in other cases
(e.g., \cite{Bshouty:exp, Ziegler:Nonarithmetic, Bertoni:pTime_RAM}) some of
the integers manipulated are arbitrary, subject only to the condition of
being sufficiently large. We refer to such arbitrary large numbers as ALNs.

Recently, in \cite{Brand:ALNs}, a general framework was proposed for the study
of the power contribution of ALNs, this being the ARAM. An
$\text{ARAM}[\textit{op\/}]$ program is defined by a
$\text{RAM}[\textit{op\/}]$ program, $r$, in the following way. The ARAM is
said to compute the function $\rho_A(\inp)$ if $r$ computes the
function $\rho(\inp,A)$, and for all $\inp$
\[
\mathop\forall\limits^\infty\!\!A \,\rho(\inp,A)=\rho_A(\inp),
\]
where $\mathop\forall\limits^\infty$ denotes ``all but a finite number''
(usually written as ``almost all'').
For the purpose of complexity calculations, the run-time associated with the
ARAM is the worst-case run-time of $r$ with the same $\inp$ over any possible
choice of $A$ (including the possibilities for which the results of $\rho$ and
of $\rho_A$ differ). If the run-time for a given $\inp$ is unbounded, over
the possible choices of $A$, the ARAM is said not to terminate.

We now introduce a new computational model which generalises the ARAM, this
being the Arbitrary Sequence RAM (ASRAM). To define the ASRAM, let us first
define the Arbitrary Large Sequence (ALS) set.

\begin{defi}[ALS]
An \emph{Arbitrary Large Sequence (ALS) set} is a nonempty set of (infinite)
integer sequences, $\mathbf{S}$, such that for any $i$ and any sequence
$\{A_k\}\in\mathbf{S}$ there exists a sequence $\{B_k\}\in\mathbf{S}$ such that
$B_i=A_i+1$, and if $j<i$, then $B_j=A_j$.
\end{defi}

The definition of the ALS set is such that if any finite list of integers
appears as a prefix of any sequence in $\mathbf{S}$, the last integer can be
increased by $1$ (and, by induction, can be replaced by any larger number),
and the result would still be a prefix of a sequence in $\mathbf{S}$.
This being the case, any finite list of integers appearing as a prefix of any
sequence in $\mathbf{S}$ remains a prefix in $\mathbf{S}$ if one extends it
by another element, given that this element exceeds some threshold value.
Other than being ``large enough'', the new element can be chosen arbitrarily.

\begin{defi}[ASRAM]
An \emph{ASRAM} is a computational model that provides the same functionality
as the RAM, but also allows calls to a pseudofunction, ``$\textit{ALN}()$'',
that returns integers.

An ASRAM is said to compute the function $f(\inp)$ if for each
$\inp$ there exists an ALS set, $\mathbf{S}$, such that for any
$\{A_i\}\in\mathbf{S}$, if the $i$'th invocation of $\textit{ALN}()$ is
replaced by the constant $A_i$ then the resulting RAM calculates
$f(\inp)$.

The run-time of the ASRAM on a given $\inp$ is the run-time of the
underlying RAM with the worst-case choice of $\{A_i\}$. (The ASRAM is taken to
be non-terminating if this worst-case is unbounded.)
\end{defi}

This definition reflects a situation where every application of
$\textit{ALN}()$ returns a number that is arbitrary other than being
sufficiently large with respect to everything that occurred in earlier steps
of the ASRAM's execution.

The ASRAM can be used to investigate a scenario in which an unbounded number
of ALNs are required. However, we can also use it for the intermediate
scenario, where only a predefined number (e.g.\ $2$) of ALNs are available to
the algorithm. This is simply done by limiting the number of times the
$\textit{ALN}()$ pseudofunction can be executed. The original ARAM is an ASRAM
limited to use $\textit{ALN}()$ at most once.

\section{Arithmetic complexity}
At face value, one may believe that multiple arbitrary numbers are no more
powerful than a single arbitrary number. However, this is not so.

In this section we show that the extra power of arbitrary sequences is present
already in the traditional arithmetic complexity model, this being the
computational model in which the basic operations used are the four arithmetic
functions, $\{+,\minus,\times,\div\}$, where
$a\minus b \defeq \max(a-b,0)$ and ``$\div$'' is integer division.
We also use ``$\bmod$'' freely in the arithmetic model, because it is an
operation straightforward to simulate using the available operations.

We stress that despite use of the name ``arithmetic'', the results presented
rely heavily on the non-arithmetic nature of ``$\div$''. In the literature
\cite{Ziegler:Nonarithmetic},
this operation is, in fact, sometimes referred to as non-arithmetic
division.

Formally stated, what we prove is the following theorem.

\begin{thm}
The class of functions that can be computed in polynomial time by an arithmetic
ASRAM is strictly larger than the class of functions that can be computed in
polynomial time by an arithmetic ARAM.
\end{thm}

\begin{proof}

Consider Algorithm~\ref{A:ASRAM}.
This algorithm utilises the fact that for an ALN, $A$, and a polynomial, $P$,
the calculation $P(A) \bmod (A-x)$ is an algorithm for computing
$P(x)$.\footnote{This is evident from the fact that in any ring, $R$, for any
elements $x,y\in R$ and any polynomial $P$ over $R$, the following holds:
$x\equiv y \pmod R \Rightarrow P(x) \equiv P(y) \pmod R$.}
Utilising this property, the algorithm calculates each $P_i$ so that
it equals $A_i^{2^{2^{(x-i)}}}$. After $\BigO(x)$ steps, it returns the value
$2^{2^{2^x}}$.

\begin{algorithm}
\caption{An arithmetic ASRAM calculating $2^{2^{2^x}}$ in $\BigO(x)$ time}
\label{A:ASRAM}
\begin{algorithmic}[1]
\For {$i\in 1,\ldots, x$}
\State $A_i\Leftarrow\textit{ALN}()$
\EndFor
\State $P_x\Leftarrow A_x\times A_x$
\For {$i\in x-1,\ldots, 1$}
\State $\textit{temp}\Leftarrow P_{i+1} \bmod (A_{i+1}-A_i)$
\State $P_i\Leftarrow P_{i+1} \bmod (A_{i+1}-\textit{temp})$
\EndFor
\State $\textit{temp}\Leftarrow P_1 \bmod (A_1-2)$
\State $\textit{rc}\Leftarrow P_1 \bmod (A_1-temp)$
\State \Return $\textit{rc}$
\end{algorithmic}
\end{algorithm}

From \cite{Bshouty:exp}, we know that an $\text{ARAM}[+,\minus,\times,\div]$
can only calculate
$2^{2^x}$ in $\Omega(\sqrt{x})$ time. The fact that Algorithm~\ref{A:ASRAM}
calculates it in $\Th(\log{x})$ time makes this an example that is polynomial
for an ASRAM but not for an ARAM.\footnote{Though not appearing in this
description, Algorithm~\ref{A:ASRAM} can be modified in a straightforward way
to calculate
$2^{2^x}$ for any $x$, and not just for $x$ values that are powers of $2$.}
\end{proof}

\section{$\text{ASRAM}[+,\leftarrow,\textit{Bool\/}]$}

In the remaining sections we consider ASRAMs with instruction sets that are
more powerful than the arithmetic operations. These include left shifting
($a \leftarrow b \defeq a\times 2^b$) and bitwise Boolean operations.
These ASRAMs will be investigated in terms of their abilities to accept
languages (rather than to calculate functions). We use the notation
$\text{$f(x)$-RAM}[\textit{op\/}]$ to denote the class of languages recognisable
by a $\text{RAM}[\textit{op\/}]$ in $f(x)$ time. (``RAM'' can be replaced
by ``ARAM'' or ``ASRAM''.) For example, $\text{$\BigO(1)$-ARAM}[\textit{op\/}]$ is
the class of languages recognisable by an $\text{ARAM}[\textit{op\/}]$ in
constant time, whereas $\text{P-ASRAM}[\textit{op\/}]$ is the class of languages
recognisable by an $\text{ASRAM}[\textit{op\/}]$ in polynomial time.

We begin by examining a case where ASRAMs afford no additional computational
power.

\begin{thm}\label{T:AS_no_div}
$\text{P-ASRAM}[+,\leftarrow,\textit{Bool\/}]=\text{PSPACE}$, where PSPACE is
the class of languages recognisable by a Turing machine (TM) working on a tape
of polynomial size.
\end{thm}

Both the RAM and ARAM working with the same operation set have been shown
\cite{Simon:division2, Brand:ALNs} to
be able to recognise PSPACE in polynomial time.

We remark that if it was known that the class of functions recognisable by
$\text{ARAM}[\textit{op\/}]$ working under some time constraints equals the
class of functions recognisable by a $\text{RAM}[\textit{op\/}]$ working under
the same time constraints, then this result would have been directly also
applicable to ASRAMs, by recursively reducing the number of ALNs required.
However, \cite{Brand:ALNs} only proves that ALNs do not add power for ARAMs
that are language recognisers. It is still possible that there are functions
that can be calculated by an $\text{ARAM}[+,\leftarrow,\textit{Bool\/}]$ but
not by a $\text{RAM}[+,\leftarrow,\textit{Bool\/}]$, and the availability of
the pseudofunction ``$\textit{ALN}()$'' may add more functions still.
Theorem~\ref{T:AS_no_div} shows that for this particular operation set, the
fact that an ARAM does not recognise more languages than a RAM carries over
also to ASRAMs. However, for results regarding language recognition, this is
by no means known to hold for a general operation set.

\begin{proof}[Proof of Theorem~\ref{T:AS_no_div}]
The proof is a direct extension of the proof of Theorem~4 in
\cite{Brand:ALNs}. In the original proof, it was shown that a polytime ARAM
with said operation set can be simulated by a PSPACE TM, given that a PSPACE
TM can determine which is larger of a pair of expressions of the
form $a\omega+b$, where $a$ and $b$ are known nonnegative integer constants
and $2^\omega$ is
the ARAM's ALN. Because $\omega$ is, by definition, ``large enough'', the
answer can be reached by simple lexicographical comparison.

In the ASRAM scenario,
we begin by picking our ALNs, $A_1, A_2,\ldots$ so as to be powers of two:
$2^{\omega_1}, 2^{\omega_2},\ldots$. We then simulate the ASRAM in exactly
the same way as was done for the ARAM.
Utilising exactly the same proof as that of
Theorem~4 in \cite{Brand:ALNs}, we conclude now that a polytime ASRAM with
said operation set can be simulated by a PSPACE TM, given that a PSPACE
TM can determine which is larger of a pair of expressions of the
form $a_0 + \sum_{i=1}^{k} a_i \omega_i$, where both $k$ and $a_0,\ldots, a_k$
are given nonnegative integers.

Once again,
because each $\omega_i$ is, per assumption, large enough compared to all
$\omega_j$ with $j<i$, a simple lexicographic comparison is enough to determine
which of two formal expressions has the larger value.

Thus, the entire simulation of the ASRAM can be performed in PSPACE.
\end{proof}

\section{$\text{ASRAM}[+,/,\leftarrow,\textit{Bool\/}]$}

In this section we prove two theorems.

\begin{thm}\label{T:O1_ASRAM}
$\text{$\BigO(1)$-ASRAM}[+,/,\leftarrow,\textit{Bool\/}]=\text{AH}$.
\end{thm}

\begin{thm}\label{T:om1_ASRAM}
$\text{$\om(1)$-ASRAM}[+,/,\leftarrow,\textit{Bool\/}]\supset\text{AH}$.
\end{thm}

In the statements of Theorems~\ref{T:O1_ASRAM} and \ref{T:om1_ASRAM},
``$/$'' is exact division (a weaker form of division, yielding the same results
as integer division, but defined only when the division is without a remainder).
AH refers to the entire arithmetic hierarchy,
$\bigcup_{i=0}^{\infty} \Sigma^0_i \cup \Pi^0_i$.

In order to prove the above, we utilise Theorem~5 of \cite{Brand:ALNs}:

\begin{thm}[\cite{Brand:ALNs}]\label{T:O1_ARAM}
Any recursively enumerable (r.e.) set can be recognised in $\BigO(1)$ time by
an $\text{ARAM}[+,/,\leftarrow,\textit{Bool}]$.
\end{thm}

\begin{proof}[Proof of Theorem~\ref{T:O1_ASRAM}]

Any ASRAM running $k$ steps
(and therefore utilising at most $k$ ALNs, $A_1,\ldots, A_k$) can be transformed
into an equivalent ASRAM that first generates $A_1,\ldots, A_k$ and then
performs any other computation. The part of the computation after the
generation of the ALNs is a RAM computation. We consider this RAM as a
language acceptor, and denote the predicate logically-equivalent to it
$\phi(\inp,A_1,\ldots, A_k)$.

The function
computed by the ASRAM is the value of $\phi$ for a sufficiently large $A_k$
given a sufficiently large $A_{k-1}$, given a sufficiently large $A_{k-2}$,
etc.. Technically, given an appropriate choice of $A_1,\ldots, A_{k-1}$, we say
that there exists some threshold, $N_k$, such that for any choice of $A_k$
satisfying $A_k>N_k$, the value of $\phi(\inp,A_1,\ldots, A_k)$ is
always the same, and it is taken to be the value of the ASRAM. In particular,
the ASRAM accepts if and only if for any choice of threshold, $N_k$, there
exists a value of $A_k$ larger than the threshold for which
$\phi(\inp,A_1,\ldots, A_k)$ is true.

This indicates that the predicate computed by the ASRAM can be formulated as
\[
\forall N_1 \exists A_1 \forall N_2 \exists A_2
\cdots \forall N_k \exists A_k
(A_1>N_1) \cap (A_2>N_2) \cap \cdots \cap (A_k>N_k) \cap
\phi(\inp,A_1,\ldots, A_k).
\]

Because
$\phi$ is calculated by a RAM, it is known to be in $\Sigma^0_1$, so the
predicate computed by the ASRAM is, by definition, in $\Pi^0_{2k}$.

We have thus established that the formula computed by the ASRAM is in AH.
We now show that any formula, ``$\phi=\exists a_1 \forall a_2 \cdots
\exists a_k \chi(\inp,a_1,\ldots, a_k)$'', with any $k$ quantifiers,
can be computed
by a constant time ASRAM with $k$ ALNs. The way to do this is to bound every
$a_i$ except the last by an ALN
$A_i$. The formula $\phi$ is logically equivalent to
``$\forall N_1 \exists A_1 \cdots
\forall N_{k-1} \exists A_{k-1}
\exists (a_1<A_1) \forall (a_2<A_2) \cdots
\exists a_k$ such that
$(A_1>N_1)\cap \cdots \cap (A_{k-1}>N_{k-1}) \cap
\chi(\inp,a_1,\ldots, a_k)$''.
(If an $a_i$ exists to satisfy some condition, its value can be bounded by some
$A_i$, whereas if
something is true for every $a_i$, it will also be true for every bounded
$a_i$.)

This derivation shows that any formula with $k$ quantifiers can be computed as
a formula with a single quantifier, given $k-1$ ALNs.

We now make use of Theorem~\ref{T:O1_ARAM}, which is equivalently stated
as $\text{$\BigO(1)$-ARAM}[+,/,\leftarrow,\textit{Bool}]\supseteq \Sigma^0_1$.

Note that an $\BigO(1)$-ARAM certainly terminates, and therefore its
underlying RAM certainly terminates as well. This means that its return
value can be inverted, making it accept a new language that is the complement
of the original accepted language. The set of complements to $\Sigma^0_1$ is
$\Pi^0_1$, from which we can conclude that
$\text{$\BigO(1)$-ARAM}[+,/,\leftarrow,\textit{Bool}]\supseteq \Sigma^0_1 \cup
\Pi^0_1$, this being the set of all formulae with a single quantifier.

We therefore first generate
$k-1$ ALNs, in a total of $\BigO(k)$ time. Given the values of the $k-1$ ALNs, the
formula to be calculated is in $\Sigma^0_1\cup\Pi^0_1$ (its one remaining
unbounded quantifier being $a_k$),
so we know it to be computable in
an additional $\BigO(1)$ time by an ARAM, for which we now use the $k$'th ALN.
In total, the ASRAM runs in $\BigO(k)$-time.

Any formula in AH is on some level of the hierarchy. To simulate it, we fix
$k$ to be that level. Thus, the simulating ASRAM runs in $\BigO(1)$-time.
\end{proof}

\begin{cor}
An ASRAM restricted to use only $k$ ALNs (regardless of its time complexity)
is equivalent to a formula on the $\Th(k)$-th level of the arithmetical
hierarchy.
\end{cor}

\begin{proof}
This result is a direct corollary of the proof for Theorem~\ref{T:O1_ASRAM},
which shows that an ASRAM utilising $k$ ALNs
can compute any formula in $\Sigma^0_k\cup\Pi^0_k$, and can be
computed by a formula in $\Pi^0_{2k}$.
\end{proof}

Consider, now, what happens when the ASRAM (not restricted to any fixed number
of ALNs), is allowed to run in $\om(1)$ time.

\begin{proof}[Proof of Theorem~\ref{T:om1_ASRAM}]
A well-known example of a function that is not in AH is ``TRUTH''. This is
a function that takes as input a formula, $\psi$, suitably encoded as an
integer, and determines whether this formula is true or not.

The inability to describe TRUTH as a formula in AH is known as Tarski's
undefinability
theorem \citep{Tarski:undefinability}. It is a corollary of G\"{o}del's
incompleteness theorem \citep{Goedel:incompleteness} and, in the formulation
given above, is a direct result of Post's theorem, stating that the
arithmetical hierarchy does not collapse \citep[see][]{Rogers:recursive}. 

Consider, first, an ASRAM working in $\Th(n)$ time, where $n=|\inp|$ is the
bit-length of its input, $\inp$. As demonstrated in
Theorem~\ref{T:O1_ASRAM}, such an ASRAM can compute, directly, any formula
with $\Th(n)$ quantifiers. Consider a formula, $\psi$, encoded in the
straightforward manner as an integer with $n$ bits. This formula will
necessarily have $\BigO(n)$ quantifiers, so a $\Th(n)$-time ASRAM (a linear-time
ASRAM) can be
used to compute its truth value. Hence, TRUTH is in
$\text{$\BigO(n)$-ASRAM}[+,/,\leftarrow,\textit{Bool\/}]$.

To extend this result from $\Th(n)$-time execution to $\om(1)$-time execution,
we note simply that the straightforward formula encoding used above may be,
perhaps, the most efficient encoding possible,
but is certainly not the only one. For example, it is possible to encode
the statement less efficiently by re-encoding the original input number,
$\inp$, as $\inp'=(2\inp+1)\times 2^T$. An arbitrary choice of $T$
allows constant-time decoding of the original $\inp$. However, because we
measure complexity as a function of the bit-length of the input,
$n'=|\inp'|$, and because the procedure shown here artificially increases this
bit-length by an arbitrarily-large value, $n-n'=T+1$, choosing a large enough
$T$ effectively reduces the run-time complexity arbitrarily. For example, if
$T$ is chosen to be $n^2$, we have $n'=\Th(n^2)$, so the ASRAM's run-time,
which is still $\Th(n)$, is merely $\Th(\sqrt{n'})$ when considered as a
function of the bit-length of its actual input, $\inp'$. With an appropriate
choice of $T$, the ASRAM's execution time, though still $\Th(n)$, can be taken
to be as low as any $\om(1)$ function of $n'$.

Tarski's undefinability theorem is independent of the exact choice of encoding
used to make the input formula, $\psi$, into a number. The new, tweaked TRUTH
function must, therefore, also lie outside of AH.

Thus, an $\om(1)$-time ASRAM can compute functions that are outside of AH.
\end{proof}

\section{Conclusions and future work}

We have fully characterised $\text{P-ASRAM}[+,\leftarrow,\textit{Bool\/}]$ and
$\text{$\BigO(1)$-ASRAM}[+,/,\leftarrow,\textit{Bool\/}]$. For
$\text{$\omega(1)$-ASRAM}[+,/,\leftarrow,\textit{Bool\/}]$, we have not
provided a full characterisation, other than stating that it is beyond AH, but
perhaps this is the best characterisation one can hope for: stratification
beyond AH is traditionally very coarse-grained. If anything, one can say that
ASRAM complexity provides us with a new and effective tool for fine-grained
stratification beyond AH.

Where future research appears most needed is regarding our result on the
arithmetic ASRAM. We have shown that
the ASRAM is a more powerful model than the ARAM under arithmetic
complexity, but full quantification of this extra power is still an interesting
open problem.

\end{document}